\documentclass[12pt]{amsart}
\usepackage{amsmath, amssymb, stmaryrd, setspace, nicefrac, multicol, amsthm, geometry}

\newtheorem{thm}{Theorem}
\newtheorem{lem}[thm]{Lemma}
\newtheorem{cor}[thm]{Corollary}
\newtheorem{prop}[thm]{Proposition}

\theoremstyle{definition}

\newtheorem{A}[thm]{Answer}

\renewcommand{\A}{\mathcal{A}}

\newcommand{\C}{\mathcal{C}}
\newcommand{\D}{\mathcal{D}}

\newcommand{\m}{\mathbf{m}}
\newcommand{\M}{\mathbf{M}}

\renewcommand{\P}{\mathcal{P}}
\newcommand{\QQ}{\mathcal{Q}}

\newcommand{\cs}{2^\omega}
\newcommand{\str}{2^{<\omega}}
\newcommand{\uh}{{\upharpoonright}}
\newcommand{\halts}{{\downarrow}}
\newcommand{\diverges}{{\uparrow}}

\renewcommand{\tt}{\mathit{tt}}
\newcommand{\KA}{\mathit{KA}}
\newcommand{\Km}{\mathit{Km}}
\newcommand{\llb}{\llbracket}
\newcommand{\rrb}{\rrbracket}
\newcommand{\hs}{\ensuremath{\emptyset'}}

\title{Bridging Computational Notions of Depth}
\author{Laurent Bienvenu, Christopher P.\ Porter}
\date{\today} 
\begin{document}

\maketitle 
\begin{abstract}
In this article, we study the relationship between notions of depth for sequences, namely, Bennett's notions of strong and weak depth, and deep $\Pi^0_1$ classes, introduced by the authors and motivated by previous work of Levin.  For the first main result of the study, we show that every member of a $\Pi^0_1$ class is order-deep, a property that implies strong depth.  From this result, we obtain new examples of strongly deep sequences based on properties studied in computability theory and algorithmic randomness.  We further show that not every strongly deep sequence is a member of a deep $\Pi^0_1$ class.  For the second main result, we show that the collection of strongly deep sequences is negligible, which is equivalent to the statement that the probability of computing a strongly deep sequence with some random oracle is 0, a property also shared by every deep $\Pi^0_1$ class.  Finally, we show that variants of strong depth, given in terms of a priori complexity and monotone complexity, are equivalent to weak depth.
\end{abstract}

\section{Introduction}

Bennett introduced the notion of logical depth in \cite{Bennett1995} as a measure of complexity, formulated in terms the amount of computation time required to reproduce a given object.  Whereas the Kolmogorov complexity of a string $\sigma\in\str$ measures the length of the shortest input given to a fixed universal machine that reproduces $\sigma$ as its output, logical depth measures the number of steps it takes to recover $\sigma$ from this shortest input.  Bennett further defined a sequence $X\in\cs$ to be strongly deep if for every computable function $t$ the logical depth of almost all of the initial segments $X\uh n$ of $X$ is greater than $t(n)$.

Bennett established several fundamental facts about strongly deep sequences, namely that the halting set $K$ is strongly deep, that no computable sequence and no Martin-L\"of random sequence is strongly deep, and that strong depth is closed upwards under truth-table reducibility (a result he referred to as the \emph{slow growth law}). Bennett further introduced the notion of weak depth, where a sequence is weakly deep if it is not truth-table reducible to a random sequence.

An analogue of deep sequences for $\Pi^0_1$ classes, i.e., effectively closed subsets of $\cs$, was developed by the present authors in \cite{BienvenuP2016}.  The authors isolated the notion of a deep $\Pi^0_1$ class as a generalized of work of Levin \cite{Levin2013}, who implicitly showed that the $\Pi^0_1$ class of consistent completions of Peano arithmetic is deep.  The basic idea, made precise in the next section, is that a $\Pi^0_1$ class $\P$ is deep if the probability of computing some length $n$ initial segment of some member of $\P$ via some Turing functional equipped with a random oracle rapidly approaches zero as $n$ grows without bound.  In \cite{BienvenuP2016}, the authors proved a number of results about deep $\Pi^0_1$ classes, including an analogue of the slow growth law for deep $\Pi^0_1$ classes in the Medvedev degrees, as well as identifying a number of examples of deep $\Pi^0_1$ classes based on properties studied in computability theory and algorithmic information theory.

The aim of this study is to show that the relationship between strongly deep sequences and deep $\Pi^0_1$ classes is no mere analogy.  In particular, we prove that every member of a deep $\Pi^0_1$ class is strongly deep, from which it follows that we gain a significant number of newly identified examples of strongly deep sequences based on results from~\cite{BienvenuP2016}.  We further show that a strongly deep sequence need not be a member of a deep $\Pi^0_1$ class.  Next, as every deep $\Pi^0_1$ class is negligible, in the sense that the probability of computing a member of such a class with a Turing functional equipped with a random oracle is zero, in light of the fact that all members of deep $\Pi^0_1$ classes are strongly deep, it is reasonable to ask   whether the collection of strongly deep sequences is negligible.  We answer this question in the affirmative, while further showing that the collection of sequences that are deep with respect to any fixed time bound is not negligible.  Finally, we consider variants of strong depth given in terms of a priori and monotone complexity and demonstrate that these two variants are equivalent to Bennett's notion of weak depth (with the latter equivalence following from work by Schnorr and Fuchs~\cite{SchnorrF1977}).  

One takeaway we aim to emphasize in this study is the importance of the slow growth law for the study of depth, akin to the role of randomness preservation in the study of algorithmic randomness.  According to the latter, every sequence that is truth-table reducible to a sequence that is random with respect to a computable measure is itself random with respect to a computable measure, which is precisely the dual of the slow growth law for deep sequences.  We anticipate that the slow growth law will continue to be a useful tool in the study of notions of depth.

The outline of the remainder of this article is as follows.  In Section~\ref{sec-background} we provide background on notions from computability theory and algorithmic randomness that we draw upon in this study.  Next, in Section~\ref{sec-slow-growth}, we provide a short proof of the slow growth law for strongly deep sequence and use it to identify some hitherto unnoticed examples of strongly deep sequences from classical computability theory.  In Section~\ref{sec-members} we prove the main result of our study, namely that every member of a deep $\Pi^0_1$ class is strongly deep.  Here we observe some consequences of this result and separate several depth notions, showing in particular that not every strongly deep sequence is a member of a deep $\Pi^0_1$ class.  Section~\ref{sec-depth-negligible} contains our second main result, that the collection of strongly deep sequences is negligible, while in Section~\ref{sec-depth-variants} we show the equivalence of depth notions given in terms of a priori and monotone complexity with Bennett's notion of weak depth.

\section{Background}\label{sec-background}

\subsection{Turing functionals}

Recall that a \emph{Turing functional} $\Phi:\subseteq\cs\rightarrow\cs$ can be defined in terms of a c.e.\ set $S_\Phi$ of pairs of strings $(\sigma,\tau)$ such that 
if $(\sigma,\tau),(\sigma',\tau')\in S_\Phi$ and $\sigma\preceq\sigma'$, then $\tau\preceq\tau'$ or $\tau'\preceq\tau$.  For each $\sigma\in\str$, we define $\Phi^\sigma$ to be the maximal string  in $\{\tau: (\exists \sigma'\preceq\sigma) (\sigma',\tau)\in S_\Phi\}$ in the order given by $\preceq$.  To obtain a map defined on $\cs$ from the set $S_\Phi$, for each $X\in\cs$, we let $\Phi^X$ be the maximal $y\in\str\cup\cs$ in the order given by $\preceq$  such that $\Phi^{X \uh n}$ is a prefix of~$y$ for all~$n$. 
\subsection{Semimeasures}

A \emph{discrete semimeasure} is a function $m:\str\rightarrow [0,1]$ satisfying $\sum_{\sigma\in\str}m(\sigma)\leq 1$.  Similarly, \emph{a continuous semimeasure} is  a function $P:\str\rightarrow [0,1]$ satisfying (i) $P(\varnothing)\leq 1$ and (ii) $P(\sigma)\geq P(\sigma0)+P(\sigma1)$ for all $\sigma\in\str$.  Given a continuous semimeasure $P$ and some $S\subseteq \str$, we set $P(S)=\sum_{\sigma\in S} P(\sigma)$.

A  discrete semimeasure $m$ is computable if its output values $m(\sigma)$ are computable uniformly in the input $\sigma\in\str$ (and similarly for continuous semimeasures). Here we will also consider \emph{lower semicomputable} semimeasures (both discrete and continuous), where a function $f:\str\rightarrow[0,1]$ is lower semicomputable if each value $f(\sigma)$ is the limit of a computable, nondecreasing sequence of rationals, uniformly in $\sigma\in\str$.

An important development due to Levin \cite{LevinZ1970} was the identification of \emph{universal} semimeasures:  for discrete semimeasures, $m$ is universal if for every lower semicomputable measure $m_0$ there is some constant $c$ such that $m_0(\sigma)\leq c\cdot m(\sigma)$.  Similarly, a continuous semimeasure $M$ is universal if for every lower semicomputable measure $P$ there is some constant $c$ such that $P(\sigma)\leq c\cdot M(\sigma)$.  Hereafter, $\mathbf{m}$ and $\mathbf{M}$ will denote fixed universal discrete and continuous semimeasures, respectively.

%
%

\subsection{Initial segment complexity}\label{subsec-isc}

Recall that the prefix-free Kolmogorov complexity of a string $\tau\in\str$ is defined by setting $K(\tau)=\min\{|\sigma|\colon U(\sigma)\halts=\tau\}$, where $U$ is a fixed universal prefix-free machine (i.e., recall that a machine $M$ is prefix-free if for $\sigma,\rho\in\str$, if $M(\sigma)\halts$ and $\sigma\prec\rho$, then $M(\rho)\diverges$).  Moreover, we can define time-bounded versions of Kolmogorov complexity.  A function $t:\omega\rightarrow\omega$ is called a \emph{time bound} if $t$ is total and non-decreasing.  Then for a fixed computable time bound $t$, the $t$-time-bounded complexity of $\tau\in\str$ is defined by setting $K^t(\tau)=\min\{|\sigma|\colon U(\sigma)\halts=\tau\text{\;in $\leq t(|\tau|)$ steps}\}$.

Note that by Levin's coding theorem, $K(\sigma)=-\log\mathbf{m}(\sigma)+O(1)$ for all $\sigma\in\str$. A similar relationship holds for computable discrete semimeasures and time-bounded Kolmogorov complexity.  First, we define a time-bounded version of $\mathbf{m}$ as follows.  As $\mathbf{m}$ is lower semicomputable, for each $s\in\omega$, we have an approximation $\mathbf{m}_s$ of $\mathbf{m}$ (i.e., for each $\sigma\in\str$, $\mathbf{m}_s(\sigma)$ is the $s$-th rational number in computable sequence that converges to $\mathbf{m}(\sigma)$).  Then given a computable time bound $t$, we set $\mathbf{m}^t(\sigma)=\mathbf{m}_{t(|\sigma|)}$, which is clearly a computable semimeasure.  We will make use of the following lemma from \cite{BienvenuDM2023} (where for functions $f,g$, $f\leq^\times g$ means that there is some $c$ such that $f(n)\leq c\cdot g(n)$ for all $n\in\omega$).

\begin{lem}[\cite{BienvenuDM2023}]{\ }\label{lem-comp-semi}
\begin{itemize}
\item[(i)] For every computable discrete semimeasure $m$, there is some computable time bound $t$ such that $m\leq^\times \mathbf{m}^t$.  
\item[(ii)] For every computable time bound $t$, $2^{-K^t}$ is a computable discrete semimeasure.
\item[(iii)] For every computable time bound $t$, there is some computable time bound $t'$ such that $2^{-K^t}\leq^\times \mathbf{m}^{t'})$.  
\end{itemize}
\end{lem}

In addition, we need the following theorem (see, e.g. \cite[Theorem 4.3(2)]{JuedesLL1994}).

\begin{thm}\label{thm-comp-semi}
For every computable time bound $t$, there is a computable time bound $t'$ such that $\mathbf{m}^t\leq^\times 2^{-K^{t'}}$.
\end{thm}

\noindent Note that by combining Lemma \ref{lem-comp-semi}(iii) and Theorem \ref{thm-comp-semi}, we obtain a resource-bounded analogue of Levin's coding theorem.

In the case of continuous semimeasures, we directly define $\KA(\sigma):=-\log \mathbf{M}(\sigma)$ to be the \emph{a priori complexity} of $\sigma\in\str$.  Just as we defined $\mathbf{m}^t$ for any computable time bound $t$, we can similarly define $\mathbf{M}^t$, which is a computable continuous semimeasure.  Moreover, we can establish the analogue of Lemma \ref{lem-comp-semi}(i):  For every computable continuous semimeasure $P$, there is some computable time bound $t$ such that $P\leq^\times  \mathbf{M}^t$.  We will also define $\KA^t:=\mathbf{M}^t$ for any given computable time bound $t$.

Lastly, we define monotone complexity in terms of monotone machines, where a monotone machine $M:\str\rightarrow\str$ satisfies the property that for $\sigma,\tau\in\mathrm{dom}(M)$, if $\sigma\preceq\tau$, then either $M(\sigma)\preceq M(\tau)$ or $M(\tau)\preceq M(\sigma)$.  Given a universal monotone machine $U$, we set $\Km(\tau)=\min\{|\tau|\colon U(\sigma)\halts\succeq\tau\}$.  Given a computable time bound, we can also define $\Km^t$ in the obvious way.

\subsection{Randomness and depth notions}
Given a computable measure $\mu$ on $\cs$ (i.e., a measure on $\cs$ where the values $\mu(\llb\sigma\rrb)$ are computable uniformly in $\sigma\in\str$), recall that a \emph{$\mu$-Martin-L\"of test} is a uniformly $\Sigma^0_1$ sequence $(U_i)_{i\in\omega}$ such that $\mu(U_i)\leq 2^{-i}$.  Recall further that a sequence $X\in\cs$ passes the test $(U_i)_{i\in\omega}$ if $X\notin\bigcap_{i\in\omega} U_i$ and $X$ is \emph{$\mu$-Martin-L\"of random} if it passes all $\mu$-Martin-L\"of tests.  In the case that $\mu$ is the Lebesgue measure on $\cs$ (which we denote by~$\lambda$), we will refer to $\lambda$-Martin-L\"of random sequences simply as Martin-L\"of random sequences.

Next,  $X\in\cs$ is \emph{strongly deep} if for every computable time bound $t$, we have $K^t(X\uh n) - K(X\uh n) \rightarrow \infty$.  A slightly stronger notion is given by order-depth, where $X\in\cs$ is \emph{order-deep} if there is a computable order function $g:\omega\rightarrow\omega$ such that $K^t(X\uh n) - K(X\uh n) \geq g(n)$ for almost every $n\in\omega$. Here we use the term `order function', or simply `order' to mean a non-decreasing and unbounded function. When $h$ is such a function, $h^{-1}(k)$ denotes the smallest $n$ such that $h(n) \geq k$. Note that $h^{-1}$ is computable when $h$ is. 

In the rest of the paper we will sometimes use an equivalent characterization of order-depth, given by the following lemma. 

\begin{lem}\label{lem-alt-order-deep} For $X\in\cs$, the following are equivalent.
\begin{itemize}
\item[(i)] $X$ is order-deep
\item[(ii)] For some computable increasing function~$h$, for any computable time bound~$t$ and almost all~$n$, $K^t(X \uh h(n)) - K(X \uh h(n)) \geq^+ n$. 
\item[(iii)] For some computable increasing function~$h$, for any computable time bound~$t$ and almost all~$n$, $\frac{\mathbf{m}(X \uh h(n))}{\mathbf{m}^t(X \uh h(n))} \geq^\times 2^n$. 
\end{itemize}
\end{lem}

The proof of this lemma is technical; for the sake of readability, we defer it to the appendix. \\

One of the key properties of strong depth is the slow growth law, given in terms of truth-table reductions.  Recall 
that a $\tt$-functional is a Turing functional that is total on all oracles; equivalently, there is a computable function $f$ such that for all $X\in\cs$, $|\Phi^{X\uh f(n)}|\geq n$.

\begin{thm}[Slow Growth Law \cite{Bennett1995}]
For $X,Y\in\cs$, if $X$ is strongly deep and $X\leq_\tt Y$, then $Y$ is strongly deep.
\end{thm}

\noindent The slow growth law also holds for order-depth.

Bennett proved that no computable sequence and no Martin-L\"of random sequence is strongly deep.  Hereafter, we will refer to sequences that are not strongly deep as being \emph{shallow}.  Bennett further showed that the halting set $\hs=\{e:\phi_e(e)\halts\}$ (where $(\phi_e)_{e\in\omega}$ is a standard enumeration of the partial computable functions) is strongly deep.

Bennett defined a weaker notion of depth:  a sequence $X\in\cs$ is \emph{weakly deep} if $X$ is not $\tt$-reducible to a Martin-L\"of random.  By the slow growth law and the fact that no Martin-L\"of random sequences are strongly deep, it follows that every strongly deep sequence is weakly deep; as shown by Bennett \cite{Bennett1995}, the converse does not hold. 
Note that it is a folklore result that a sequence is not truth-table reducible to a Martin-L\"of random sequence if and only if it is not random with respect to a computable measure, thereby providing an alternative characterization of weak depth.

\subsection{Deep $\Pi^0_1$ classes and negligiblity}\label{subsec-depth-neg}

As noted in the introduction, the authors in \cite{BienvenuP2016} introduced the notion of a deep $\Pi^0_1$ class as the abstraction of a phenomenon first isolated by Levin in \cite{Levin2013} Given a $\Pi^0_1$ class $\P$, recall that there is a canonical co-c.e.\ tree $T\subseteq\str$ such that $\P=[T]$, i.e., $\P$ is the collection of all infinite paths through $T$; more specifically, this tree $T$ is the set of all initial segments of members of $\P$.  For $n\in\omega$, let $T_n$ be the set of all strings in $T$ of length $n$.  We say that a $\Pi^0_1$ class $\P$ is \emph{deep} if there is some order $g$ such that $\M(T_n)\leq 2^{-g(n)}$.  Equivalently, $\P$ is deep is there is some order $h$ such that $\M(T_{h(n)})\leq 2^{-n}$.

An analogue of the slow growth law holds for deep $\Pi^0_1$ classes in a suitable degree structure, namely the strong degrees (also referred to as the Medvedev degrees).  Given $\Pi^0_1$ classes $\P$ and $\QQ$, we say that $\P$ is strongly reducible to $\QQ$, written $\P\leq_s\QQ$, if there is some Turing functional $\Phi$ such that for every $Y\in \QQ$, there is some $X\in\P$ such that $X=\Phi(Y)$; equivalently, we have $\Phi(\QQ)=\P$.  As noted in \cite{BienvenuP2016}, we can assume here that $\Phi$ is a $\tt$-functional, a fact that will be useful in this study.  Then we have:

\begin{thm}[Slow Growth Law for $\Pi^0_1$ classes, \cite{BienvenuP2016}]
For $\Pi^0_1$ classes $\P,\QQ\subseteq\cs$, if $\P$ is deep and $\P\leq_s \QQ$, then $\QQ$ is deep.
\end{thm}

Depth for $\Pi^0_1$ classes implies a property that holds more broadly for subsets of $\cs$, namely the property of being negligible. 
First, observe that a lower semicomputable semimeasure $P$ can be trimmed back to a measure $\overline P\leq P$ (see \cite{BienvenuHPS2014} details).  In particular, we can trim back the universal lower semicomputable semimeasure
 $\mathbf{M}$ to get a measure $\mathbf{\overline{M}}$.  One key result concerning $\mathbf{\overline{M}}$ is that for a measurable set $\A\subseteq \cs$, $\mathbf{\overline{M}}(\A)=0$ if and only if $\lambda(\{X: (\exists Y\in\A)\; Y\leq_T X\})=0$; that is, from the point of view of Lebesgue measure, only relatively few sequences can compute of member of~$\A$. Following Levin (see for example~\cite{Levin1984}), we call such sets $\A$ \emph{negligible}.  As we can equivalently consider the collection of random sequences that compute a member of $\A$, we can recast negligibility in terms of probabilistic computation:  a collection~$\A$ is negligible if the probability of probabilistically computing a member of $\A$ is zero.  Note that every deep $\Pi^0_1$ class is thus negligible; in fact, we can interpret the property of depth for a $\Pi^0_1$ class $\P$ as the property that the probability of computing the first~$n$ bits of a member of $\P$ converges to 0 effectively in $n\in\omega$.  As shown in \cite{BienvenuP2016}, not every negligible $\Pi^0_1$ class is deep.
 
Two other notions related to depth and negligibility for $\Pi^0_1$ classes studied in \cite{BienvenuP2016} are the notions of $\tt$-depth and $\tt$-negligibility:
\begin{itemize}
\item A $\Pi^0_1$ class $\P$ with canonical co-c.e.\ tree $T$ is \emph{$\tt$-deep} if for every computable measure $\mu$ there is some computable order $g$ such that $\mu(T_n)\leq 2^{-g(n)}$
for all $n\in\omega$, or equivalently, for every computable measure $\mu$, there is some computable order $h$ such that $\mu(T_{h(n)})\leq 2^{-n}$.

\item A measurable set $\C\subseteq\cs$ is \emph{$\tt$-negligible} if $\mu(\C)=0$ for every computable measure $\mu$, or equivalently, $\lambda(\{X\in\C\colon\exists Y\in\cs \;X\leq_\tt Y\})=0$.
\end{itemize}
Unlike the notions of depth and negligibility, we have the following equivalence:
\begin{thm}[\cite{BienvenuP2016}]\label{thm-tt-depth}
Let $\P\subseteq\cs$ be a $\Pi^0_1$ class.  The following are equivalent:
\begin{itemize}
\item[(i)] $\P$ is $\tt$-deep.
\item[(ii)] $\P$ is $\tt$-negligible. 
\item[(iii)] For every computable measure~$\mu$ on $\cs$, $\P$ contains no $\mu$-Martin-L\"of random element. 
\end{itemize}
\end{thm}
Note that by the alternative characterization of weak depth discussed in the previous subsection, Theorem \ref{thm-tt-depth}, a $\Pi^0_1$ class is $\tt$-deep if and only if all of its members are weakly deep.

\section{On the Slow Growth Law}\label{sec-slow-growth}

In this section, we provide a proof of the slow growth law and provide some hitherto unobserved consequences of the result.  In particular, the proof of the slow growth law that we offer here is distinct from others in the literature in two respects.  First, unlike other proofs in the literature, such as the one found in \cite{JuedesLL1994}, which are more complexity-theoretic (using the machinery of Kolmogorov complexity), our proof is measure-theoretic, being based on computable semimeasures.  Second, the proof offered here is much more direct than currently available proofs of the slow growth law.

We begin with a few words to simplify the setting. First, as noted in the previous section $X\in\cs$ is strongly deep when 
\[
K^t(X \uh n) - K(X \uh n) \rightarrow \infty
\]
for every computable time bound~$t$, but a more pleasant way to rephrase this, by elevating to the power of~$2$ on both sides, is to require that
\[
\frac{\m(X \uh n)}{p(X \uh n)} \rightarrow \infty
\]
for every computable discrete semi-measure~$p$ (that the two phrasings are equivalent follows directly from Lemma~\ref{lem-comp-semi} and Theorem~\ref{thm-comp-semi}). 

Second, if $\Gamma$ is a tt-functional with use $\gamma$, we can naturally extend $\Gamma$ to strings by defining $\Gamma(\sigma)$ to be the first~$n$ bits of $\Gamma^\sigma$, where~$n$ is such that $|\sigma| \in [\gamma(n), \gamma(n+1))$. Seen as a function on strings, $\Gamma$ has two properties that we will need:
\begin{itemize}
\item[(i)] $\Gamma$ is total;
\item[(ii)] For every string~$\tau$, the preimage $\Gamma^{-1}(\tau)$ is a finite set that can be computed uniformly in~$\tau$. 
\end{itemize}

We are now ready to prove our main theorem.

\begin{thm}
Let $F$ be total computable function on strings such that for all~$\tau$, $F^{-1}(\tau)$ is a finite set which can be computed uniformly in~$\tau$. Let $p$ be a computable discrete semimeasure. Then, there exists a computable discrete semimeasure~$q$ such that 
\[
\frac{\m(F(\sigma))}{q(F(\sigma))} \leq^\times \frac{\m(\sigma)}{p(\sigma)}
\]
\end{thm}

It is now clear that this theorem implies the slow growth law with $F$ identified with the extension of $\Gamma$ to $\str$, setting $\sigma=X \uh n$, and letting~$n$ tend to $\infty$. Let us prove the theorem.\\

\begin{proof}
Define~$q$ simply as the push-forward measure of $p$ under~$F$:
\[
q(\tau) = \sum_{\sigma \in F^{-1}(\tau)} p(\sigma)
\]
Our assumption on $F$ and $p$ imply that $q$ is computable. It is a semimeasure as $\sum_\tau q(\tau) = \sum_\tau \sum_{\sigma \in F^{-1}(\tau)} p(\sigma) = \sum_\sigma p(\sigma) \leq 1$.\\

 Now, define for all~$\sigma$:
\[
m(\sigma) = \frac{\m(F(\sigma))}{q(F(\sigma))} \cdot p(\sigma) 
\]
We claim that $m$ is a lower semicomputable discrete semimeasure. That $m$ is lower semicomputable is clear, since $\m$ is lower semicomputable and $p, q,$ and  $F$ are computable. That $m$ is a discrete semimeasure can be established as follows:
\[
\sum_\sigma m(\sigma) = \sum_\tau \sum_{\sigma \in F^{-1}(\tau)} m(\sigma) = \sum_{\tau} \sum_{\sigma \in F^{-1}(\tau)} \frac{\m(F(\sigma))}{q(F(\sigma))} \cdot p(\sigma) \]
\[
 = \sum_\tau \frac{\m(\tau)}{q(\tau)} \cdot \left( \sum_{\sigma \in F^{-1}(\tau)} p(\sigma) \right) = \sum_\tau \frac{\m(\tau)}{q(\tau)} \cdot q(\tau) = \sum_\tau \m(\tau) \leq 1.
\]
By maximality of $\m$, we have $m(\sigma) \leq^\times \m(\sigma)$, which by definition of~$m$ gives the desired inequality. 
\end{proof}

We note in passing that our proof also implies the slow growth law for order-depth: indeed if in the above we have $\frac{\m(\Gamma^X \uh n)}{q(\Gamma^X \uh n)} \geq h(n)$ for some computable order~$h$, then $\frac{\m(X \uh \gamma(n))}{p(X \uh \gamma(n))} \geq h(n)$ where $\gamma$ is the use of $\Gamma$, which by Lemma~\ref{lem-alt-order-deep} shows that $X$ is order-deep.

We note here some previously unnoticed consequence of slow growth law.  First,  observe that the standard unsolvable problems from computability theory are strongly deep, including:
\begin{itemize}
\item Fin = $\{x: W_x\text{ is finite}\}$
\item Inf =$\{x: W_x\text{ is infinite}\}$
\item Tot = $\{x: \phi_x\text{ is total}\}$
\item Cof = $\{x: W_x\text{ is cofinite}\}$
\item Comp = $\{x: W_x\text{ is computable}\}$
\item Ext $\{x: \phi_x\text{ is extendible to a total computable function}\}$
\end{itemize}
Indeed, for each such class $C$, we have $\hs \leq_1C$, i.e., there is a computable function $f:\omega\rightarrow\omega$ such that $n\in \hs$ if and only if $f(n)\in C$ for all $n\in\omega$.  Since every 1-reduction defines a $\tt$-functional, the result follows from the slow growth law and the fact that $\hs$ is strongly deep.  

In fact, we can strengthen this observation for \emph{all} non-trivial index sets, thereby strengthening Rice's theorem in terms of strong depth:

\begin{thm}
If $C\subseteq\omega$ is a shallow index set, i.e., $C$ is not strongly deep, then either $C=\emptyset$ or $C=\omega$.
\end{thm}

\noindent This follows immediately from the classical version of Rice's theorem, for if $C$ is a non-trivial index set, then either $\hs \leq_1 C$ or $\overline \hs \leq _1C$.  Since $\overline \hs$ is strongly deep by the slow growth law, the result is clear.

We will see additional applications of the slow growth law in the remaining sections.

\section{Members of deep $\Pi^0_1$ classes}\label{sec-members}

When deep $\Pi^0_1$ classes were defined in \cite{BienvenuP2016}, the authors referred to the notion as a type of depth in analogy with Bennett's original notion of logical depth (as, for instance, an analogue of the slow growth law for deep $\Pi^0_1$ classes was established in \cite{BienvenuP2016}).  We now show that the connection between these two depth notions is much closer than merely satisfying an analogy, as we prove that the members of deep $\Pi^0_1$ classes are strongly deep; in fact, we prove the stronger result that all such members are order-deep.   Note that this is analogous to the result that every member of a $tt$-deep $\Pi^0_1$ class is weakly deep, a consequence of Theorem \ref{thm-tt-depth} discussed at the end of Section \ref{subsec-depth-neg}.  

\begin{thm}\label{thm-deep-members}
Every member of a deep $\Pi^0_1$ class $\P$ is order-deep. 
\end{thm}

\begin{proof}
Let $T$ be the canonical co-c.e.\ tree associated to the $\Pi^0_1$ class $\P$. Let $h$ be a computable order such that 
\[
\sum_{\sigma \in T_{h(n)}} \M(\sigma) \leq 2^{-2n}.
\]
for all~$n$. 

 Let $t$ be a computable time bound. By virtue of the inequalities $\M \geq^\times \m \geq^\times \m^t$, the above inequality implies
\[
\sum_{\sigma \in T_{h(n)}} \m^t(\sigma) \leq 2^{-2n}.
\]
Since $\m^t$ is computable and $T$ is co-c.e., one can effectively compute a sequence $s_n$ such that 
\[
\sum_{\sigma \in T_{h(n)[s_n]}} \m^t(\sigma) \leq 2^{-2n}.
\]

Let now $p$ be the computable discrete semimeasure defined by $p(\tau) = 2^n \cdot \m^t(\tau)$ if~$\tau$ belongs to $T_{h(n)}[s_n]$ for some~$n>0$, and $p(\tau)=0$ otherwise. That $p$ is computable is clear from the definition (and the computability of the sequence $s_n$), and that it is a semimeasure follows from
\[
\sum_\tau p(\tau) = \sum_{n>0} \sum_{\tau \in T_{h(n)}[s_n]} 2^n\cdot \m^t(\tau) \leq \sum_{n>0} 2^n \cdot 2^{-2n} \leq 1.
\]

Now, if $X$ is a member of~$\P$, that is, $X$ is a path through~$T$,  then for each $n>0$, we have $X\uh h(n) \in T_{h(n)}$ and thus
\[
p(X\uh h(n)) = 2^n \cdot \m^t(X \uh h(n)).
\]
As $\m \geq^\times p$ (because $p$ is a discrete semi-measure), we have
\[
\m(X\uh h(n)) \geq^\times 2^n \cdot \m^t(X \uh h(n))
\] 
By Lemma~\ref{lem-alt-order-deep}, we can conclude that~$X$ is order-deep.
\end{proof}

One immediate consequence of Theorem \ref{thm-deep-members} is the following.

\begin{cor}\label{cor-deep-tree}
If $\P$ is a deep $\Pi^0_1$ class, then the canonical co-c.e.\ tree $T$ associated with $\P$ is order-deep.
\end{cor}

\begin{proof}
Let $X\in\P$ be the leftmost path through $\P$, which is order-deep by Theorem~\ref{thm-deep-members}. Then if $T$ is the canonical co-c.e.\ tree associated with $\P$, we have $X\leq_{tt} T$, so by the slow growth law, $T$ is order-deep.
\end{proof}

The converse of this result does not hold.  

\begin{thm}\label{thm-deep-tree-counterexample}
There is an order-deep co-c.e.\ tree $T$ such that $[T]$ is not a deep $\Pi^0_1$ class.
\end{thm}

\begin{proof}
Let $S\subseteq\str$ be any order-deep co-c.e.\ tree and let $T=0^\frown S\cup1^\frown\str$.
Clearly $T$ is a co-c.e.\ tree because $S$ is.  To see that $T$ is order-deep, note that $\sigma\in S$ if and only if either (i) $\sigma$ is the empty string or (ii) $0 \preceq \sigma$ and $\sigma \in T$. It thus follows that $S\leq_{tt} T$, and so by the slow growth law, $T$ is order-deep.  Finally, as the sequence $1^\omega$ is a member of the $\Pi^0_1$ class $[T]$ and every member of a deep $\Pi^0_1$ class must be order-deep by Theorem \ref{thm-deep-members}, it follows that $[T]$ cannot be deep.

\end{proof}

A similar pair of results hold for $\tt$-deep $\Pi^0_1$ classes. 

\begin{thm}
If $\P$ is a $\tt$-deep $\Pi^0_1$ class, then the canonical co-c.e.\ tree $T$ associated with $\P$ is weakly deep.
\end{thm}

\begin{proof}
Given a $\Pi^0_1$ class $\P$ with canonical co-c.e.\ tree $T$, suppose that $T$ is not weakly deep.  Then $T$ is Martin-L\"of random with respect to some computable measure.  
As in the proof of Corollary \ref{cor-deep-tree}, if $X\in\P$ is the leftmost path through $\P$, then have $X\leq_{tt} T$.  By randomness preservation (according to which if $Y$ is Martin-L\"of random with respect to a computable measure and $X\leq_{\tt}Y$, then $X$ is Martin-L\"of random with respect to a computable measure; see, e.g., \cite{BienvenuP2012}), $X$ is Martin-L\"of random with respect to a computable measure, and hence $\P$ is not $\tt$-deep.
\end{proof}

\begin{thm}
There is a weakly deep co-c.e.\ tree $T$ such that $[T]$ is not a $\tt$-deep $\Pi^0_1$ class.
\end{thm}

\begin{proof}
The proof is nearly identical to that of Theorem \ref{thm-deep-tree-counterexample}. Given any  weakly deep co-c.e.\ tree $T_0$, let $T=0^\frown T_0\cup1^\frown\str$.  Since $T_0\leq_\tt  T$, it follows that $T$ must be weakly deep.  However, $[T]$ contains every Martin-L\"of sequence that begins with a 1, it cannot be $\tt$-deep.
\end{proof}

Note earlier that we stated the result that being a $\tt$-deep $\Pi^0_1$ class not only implies that every member of the class is weakly deep, but also that this latter condition is equivalent to the property of being $\tt$-deep.  This equivalence does not hold in the case of deep $\Pi^0_1$ classes:  

\begin{thm}
There is a $\Pi^0_1$ class $\P$ such that (i) every $X\in\P$ is strongly deep but (ii) $\P$ is not deep.
\end{thm}

\begin{proof}
As shown in \cite[Theorem 4.7]{BienvenuP2016}, given any deep $\Pi^0_1$ class $\QQ$, there is a non-deep $\Pi^0_1$ class $\P$ such that every member of $\P$ is a finite modification of some member of $\QQ$.  By the slow growth law, order-depth is invariant under finite modifications, so every member of any such class $\P$ is  strongly deep.
\end{proof}

\subsection{Additional examples of deep sequences}

Theorem \ref{thm-deep-members} also allows us to derive a number of examples of deep sequences.
In \cite{BienvenuP2016} it was shown that the following collections of sequences form deep $\Pi^0_1$ classes:
\begin{enumerate}
\item the collection of consistent completions of Peano arithmetic;
\item the collection of ($\alpha, c$)-shift-complex sequences for computable $\alpha\in(0,1)$ and $c\in\omega$, where a sequence $X\in\cs$ is ($\alpha, c$)-shift-complex if $K(\tau)\geq \alpha|\tau|-c$ for every substring $\tau$ of $X$;
\item the collection of $\mathrm{DNC}_q$ functions with $\sum_{n\in\omega}\frac{1}{q(n)}=\infty$, where $f:\omega\rightarrow\omega$ is a $\mathrm{DNC}_q$ function if $f$ is total function such that $f(n)\neq \phi_n(n)$  and $f(n)<q(n)$ for all $n\in\omega$;
\item the collection of codes of infinite sequences of finite sets $(F_0,F_1,\dotsc)$ of strings of maximal complexity, i.e., there are computable functions $\ell,f,d:\omega\rightarrow\omega$ such that for all $n\in\omega$, (i) $|F_n|=f(n)$, (ii) $|\sigma|=\ell(n)$ for $\sigma\in F_n$, and (iii) $K(\sigma)\geq\ell(n)-d(n)$ for $\sigma\in F_n$; and
\item the collection of codes of $K$-compression functions with constant $c$, where $g:\str\rightarrow\omega$ is a $K$-compression function with constant $c$ if (i) $g(\sigma)\leq K(\sigma)+c$ for all $\sigma\in\str$ and (ii) $\sum_{\sigma\in\str}2^{-g(\sigma)}\leq 1$.

\end{enumerate}

As an immediate consequence of Theorem \ref{thm-deep-members} and the above results from \cite{BienvenuP2016}, we have:

\begin{cor} Every sequence in the following collections is strongly deep:
\begin{enumerate}
\item the collection of consistent completions of Peano arithmetic;
\item the collection of shift-complex sequences (i.e., the sequences that are $(\alpha,c)$-shift complex for some computable $\alpha\in(0,1)$ and $c\in\omega$);
\item the collection of codes of $\mathrm{DNC}_q$ functions with $\sum_{n\in\omega}\frac{1}{q(n)}=\infty$;
\item the collection of codes of infinite sequences of finite sets of strings of maximal complexity; and
\item the collection of codes of $K$-compression functions (i.e., $K$-compression functions with constant $c$ for some $c\in\omega$).
\end{enumerate}
\end{cor}

We can obtain further examples of members of deep $\Pi^0_1$ classes using a version of the slow growth law for members of deep $\Pi^0_1$ classes.

\begin{lem}\label{lem-deep-tt}
If $X$ is a member of a deep $\Pi^0_1$ class and $X\leq_{tt}Y$, then $Y$ is a member of a deep $\Pi^0_1$ class.
\end{lem}

\begin{proof}
Let $\P$ be a $\Pi^0_1$ class contain $X$, and let $\Phi$ be a total Turing functional satisfying $\Phi(Y)=X$.  Then by the slow growth law for deep $\Pi^0_1$ classes, $\Phi^{-1}(\P)$ is a deep $\Pi^0_1$ class that contains $Y$, which must be strongly deep by Theorem \ref{thm-deep-members}.
\end{proof}

\begin{thm}{\ }
\begin{itemize}
\item[(i)] The halting set $\emptyset'=\{e\in\omega\colon\phi_e(e)\halts\}$ is a member of a deep $\Pi^0_1$ class.
\item[(ii)]For every $X\in\cs$, $X'=\{e\colon \phi^X_e(e)\halts\}$ is a member of a deep $\Pi^0_1$ class.
\item[(iii)] Every non-trivial index set is a member of a deep $\Pi^0_1$ class.
\end{itemize}
\end{thm}

\begin{proof}
(i) There is a $\mathrm{DNC}_2$ function $f$ such that $f\leq_{tt}\emptyset'$ (see \cite[Remark 1.8.30]{Nies2009}).  Since the collection of $\mathrm{DNC}_2$ functions forms a deep $\Pi^0_1$ class, the result follows from Lemma \ref{lem-deep-tt}.\\
\noindent (ii) Since $\emptyset'\leq_{tt}X'$ for every $X\in\cs$ and $\emptyset'$ is a member of a deep $\Pi^0_1$ class by part (i), the result follows from Lemma \ref{lem-deep-tt}.\\
\noindent (iii). By Rice's theorem, every non-trivial index set $C$ satisfies $\emptyset'\leq_1 C$ or $\overline {\emptyset'}\leq_1 C$, and so the result follows again from part (i) and Lemma \ref{lem-deep-tt}.

\end{proof}

\subsection{Separating depth notions}

In light of Theorem \ref{thm-deep-members}, it is natural to ask whether every order-deep sequence is a member of a deep $\Pi^0_1$ class.  We show that this does not hold by establishing several propositions of independent interest.

Recall that a sequence $X$ is complex if there is some computable order $g$ such that $K(X\uh n)\geq g(n)$.  As shown explicitly in \cite{HolzlP2017}, one can equivalently define a sequence to be complex in terms of a priori complexity, i.e., $X$ is complex if and only if there is some computable order $h$ such that $KA(X\uh n)\geq h(n)$.  We use this second characterization to derive the following:

\begin{prop}\label{prop-deep-complex}
Every member of a deep $\Pi^0_1$ class is complex.
\end{prop}

\begin{proof}
Let $\P$ be a deep $\Pi^0_1$ class with associated co-c.e.\ tree $T$.  Then there is some computable order $h:\omega\rightarrow\omega$  such that $\mathbf{M}(T_n)\leq 2^{-h(n)}$.  Given $X\in\P$, since $X\uh n\in T_n$, we have $\mathbf{M}(X\uh n)\leq\mathbf{M}(T_n)\leq 2^{-h(n)}$.  Taking the negative logarithm yields $\mathit{KA}(X\uh n)\geq h(n)$, from which the conclusion follows.
\end{proof}

Next, we have:

\begin{prop}\label{prop-incomplete-nonmember}
No sequence that is Turing equivalent to an incomplete c.e.\ set is a member of a deep $\Pi^0_1$ class.
\end{prop}

\begin{proof}
Suppose that $X$ is a member of a deep $\Pi^0_1$ class and is Turing equivalent to some incomplete c.e.\ set $Y$.  By Proposition \ref{prop-deep-complex}, $X$ is complex.  It follows from work of Kjos-Hanssen, Merkle, and Stephan \cite{KjosHanssenMS2011} that $X$ has DNC degree.  But then $Y$ is an incomplete c.e.\ set of DNC degree, which contradicts Arslanov's completeness criterion (see, e.g., \cite[Theorem 4.1.11]{Nies2009}).
\end{proof}


\begin{thm}\label{thm-order-deep-nonmember}
There is an order-deep sequence that is not complex (hence is not a member of any deep $\Pi^0_1$ class).
\end{thm}

\begin{proof}
In~\cite{JuedesLL1994}, Juedes, Lathrop, and Lutz introduced the notion of weakly useful sequence (we do not recall the definition of this notion here and refer the reader to their paper) and showed that (i) every weakly useful sequence is order-deep and (ii) every high degree contains a weakly useful sequence. Our theorem then follows: Let $X$ be high, incomplete, and weakly useful (hence order-deep). By the same reasoning in the proof of Proposition~\ref{prop-incomplete-nonmember}, $X$ is not complex. 
\end{proof}

We note another consequence of Proposition \ref{prop-incomplete-nonmember}, namely  that the leftmost path of every deep $\Pi^0_1$ class is Turing complete.  Indeed, the leftmost path of a $\Pi^0_1$ class has c.e.\ degree and thus must be Turing complete by Proposition \ref{prop-incomplete-nonmember}.

Having separated order-depth from being a member of a deep $\Pi^0_1$ class, we can use a similar line of reasoning to further separate order-depth from Bennett's original notion of depth.  We need one auxiliary result.
 
\begin{thm}[Moser, Stephan \cite{MoserS2017}]\label{thm-moser-stephan}
Every order-deep sequence is either high or of DNC degree.
\end{thm} 

\begin{thm}
There is a strongly deep sequence that is not order-deep.
\end{thm}

\begin{proof}
Downey, MacInerny, and Ng  \cite{DowneyMN2017} constructed a low, deep sequence $A$ of c.e.\ degree.  As $A$ can neither be high nor of DNC degree (as it is incomplete), it follows from Theorem \ref{thm-moser-stephan} that $A$ is not order-deep. 
\end{proof}

We have seen by Theorem~\ref{thm-deep-members} that members of deep $\Pi^0_1$ classes are order-deep, and by Proposition~\ref{prop-deep-complex} that they are complex. We end this section by showing that these two properties alone are not enough to characterize members of deep $\Pi^0_1$ classes. 

\begin{thm}\label{thm:complex-order-deep-not-member}
There exists a sequence $X$ which is complex, order-deep, and not a member of any deep $\Pi^0_1$ class. 
\end{thm}

We will need the following auxiliary lemma of independent interest.  Recall that the join $Y \oplus Z$ of two sequences $Y$ and $Z$ is the sequence obtained by interleaving their bits: $Y \oplus Z = Y(0)Z(0)Y(1)Z(1) \ldots$. Similarly, for two strings $\sigma$ and $\tau$ of the same length we can define $\sigma \oplus \tau$ in the same way. 

\begin{lem}\label{lem-ae-deep}
If a sequence $Y$ is not a member of any deep $\Pi^0_1$ class, then for almost every $Z$ (in the sense of Lebesgue measure), $Y \oplus Z$ is not a member of any deep $\Pi^0_1$ class.
\end{lem}

\begin{proof}
We prove this lemma by contrapositive. Suppose that $Y$ is such that for positive measure many sequences $Z$, $Y \oplus Z$ is a member of a deep $\Pi^0_1$ class. If this is so, as there are only countably many deep $\Pi^0_1$ classes, this means that there is a fixed deep $\Pi^0_1$ class~$\mathcal{C}$ such that with probability $>\delta$ over $Z$ (with $\delta$ a positive rational), we have that $Y \oplus Z$ belongs to~$\mathcal{C}$. Consider the $\Pi^0_1$ class $\mathcal{D}$ consisting of sequences $A$ such that for any~$n$, there are at least $\delta \cdot 2^n$ strings $\tau$ such that $(A \uh n) \oplus \tau$ is in the canonical tree~$T$ of~$\mathcal{C}$. By definition, $Y$ belongs to $\mathcal{D}$. We claim that $\D$ is deep, which will prove the lemma. Since $\mathcal{C}$ is deep, there is a computable order $h$ such that $\M(T_{2n}) < 1/h(n)$ for all~$n$. Let $P$ be the continuous semimeasure defined on strings of even length by $P(\sigma \oplus \tau) = \M(\sigma) \lambda(\tau)$. If $S$ is the canonical co-c.e.\ tree of $\mathcal{D}$, then by definition of $\mathcal{D}$, we have 
\[
P(T_{2n}) \geq \M(S_n) \cdot \delta.
\]
By universality of $\M$, we also have $\P(T_{2n}) \leq^\times \M(T_{2n}) < 1/h(n)$. Putting these two inequalities together, it follows that 
\[
\M(S_n) \leq^\times \frac{1}{\delta \cdot h(n)}
\]
which shows $\mathcal{D}$ is a deep $\Pi^0_1$ class. 
\end{proof}

\begin{proof}[Proof of Theorem \ref{thm:complex-order-deep-not-member}]
Having proven Lemma \ref{lem-ae-deep}, take now $Y$ a sequence that is order-deep and not a member of any deep $\Pi^0_1$ class (whose existence was established in Theorem~\ref{thm-order-deep-nonmember}). Pick $Z$ at random and form the sequence $X = Y \oplus Z$. With probability~$1$ over $Z$:
\begin{itemize}
\item $X$ is complex. Indeed $K(X \uh 2n) \geq^+ K(Z \uh n) \geq^+ n$ by the Levin-Schnorr theorem.
\item $X$ is not a member of any deep $\Pi^0_1$ class by Lemma \ref{lem-ae-deep}.
\end{itemize}

Moreover, regardless of the value of $Z$, $X$ $\tt$-computes $Y$ which is order-deep, hence by the slow growth law for order-deep sequences, $X$ is itself order-deep. These three properties tell us that with probability~$1$ over $Z$, $X = Y \oplus Z$ is as desired. 

\end{proof}

\section{Strong Depth is Negligible}\label{sec-depth-negligible}

As observed in \cite{BienvenuP2016} every deep $\Pi^0_1$ class is negligible.  Since the collection of strongly deep sequences forms a strict superclass of the collection of members of all deep $\Pi^0_1$ classes by Theorems \ref{thm-deep-members} and  \ref{thm-order-deep-nonmember}, it is natural to ask whether the collection of strongly deep sequences is negligible, which we answer in the affirmative.

\begin{thm}\label{thm-deep-neg}
The class of strongly deep sequences is negligible. 
\end{thm}

\begin{proof}
For the sake of contradiction, assume there exists a functional $\Phi$ such that 
\[
\lambda \{X \ : \ \Phi^X \text{~ is deep ~}\}  > 0.9
\]
(where we choose this latter value without loss of generality by the Lebegue Density Theorem).
Let $p$ be a computable semimeasure such that $\liminf_n \frac{\m(n)}{p(n)} < \infty$. The existence of such a $p$ follows from the existence of \emph{Solovay functions} (see~\cite{BienvenuDNM2015}), which are functions $f \geq^+ K$ such that $\liminf_n f(n)-K(n)<\infty$. Setting $p(n) = 2^{-f(n)-c}$ with $f$ a Solovay function and $c$ large enough gives us the properties of $p$ we need.

We now define a computable discrete semimeasure as follows. For every~$n\in\omega$, effectively find a family of clopen sets $\{C_\sigma \ : \  |\sigma|=n\}$ such that $\Phi^X \succeq \sigma$ for all~$X \in C_\sigma$ and $\sum_{|\sigma|=n} \lambda(C_\sigma) >0.9$. Then, set for all~$\sigma$ of length~$n$:
\[
q(\sigma) = \lambda(C_\sigma) \cdot p(n)
\]
It is clear that $q$ is computable.  Moveover, $q$ is a discrete semimeasure, since 
\[
\sum_{\sigma\in\str} q(\sigma) = \sum_{n\in\omega} \sum_{|\sigma|=n} \lambda(C_\sigma) \cdot p(n) = \sum_{n\in\omega} p(n) \sum_{\sigma\in\str\colon|\sigma|=n} \lambda(C_\sigma) \leq  \sum_{n\in\omega} p(n) \leq 1.
\]

For any~$Y$ that is strongly deep, we must have $\frac{\m(Y \uh n)}{q(Y \uh n)} \rightarrow \infty$. For all $(n,d)$, define
\[
B^d_n = \{\sigma\ : \  |\sigma| = n ~\text{~ ~ and ~ ~ }~ \m(\sigma) > d \cdot q(\sigma)\}
\]
\noindent By our hypothesis on the functional~$\Phi$, this means that for every constant~$d$, for almost all~$n$, $\lambda (\{X \ : \  \Phi^X \uh n  \in B^d_n \})  > 0.9$.

Now, consider the quantity $\sum_{\sigma \in B^d_n} \lambda(C_\sigma)$. On the one hand,

\begin{align*}
 \sum_{\sigma \in B^d_n} \lambda(C_\sigma) & = \sum_{\sigma \in B^d_n}  \frac{q(\sigma)}{p(n)} \hspace{1cm}\text{(by definition of $q$)}\\
  & \leq \sum_{\sigma \in B^d_n}  \frac{\m(\sigma)}{d \cdot p(n)} \hspace{1cm}\text{(by definition of $B^d_n$)}\\
  & \leq \frac{1}{d \cdot p(n)} \sum_{|\sigma|=n} \m(\sigma)\\
  & \leq \frac{\m(n) \cdot O(1)}{d \cdot p(n)} \hspace{1cm} \text{(using the identity $\sum_{|\sigma|=n} \m(\sigma) =^\times \m(n)$)}\\
 \end{align*}

On the other hand, for almost all~$n$:

\begin{align*}
 \sum_{\sigma \in B^d_n} \lambda(C_\sigma) & \geq \lambda(\{ X \ : \ \Phi^X \uh n \in B^d_n  \}) -0.1 \hspace{1cm}\text{(because $\lambda(\bigcup_{|\sigma|=n} C_\sigma) > 0.9$)}\\
   & \geq 0.9 - 0.1 \\
   & \geq 0.8
 \end{align*}

Putting the two together, we have established that for all~$d$, for almost all~$n$, $\frac{\m(n)}{p(n)} > d/O(1)$, i.e., $\lim_n \frac{\m(n)}{p(n)} = \infty$. This contradicts the choice of $p$.
\end{proof}

Note, by contrast, that the collection of weakly deep sequences is not negligible.  Indeed, as shown by Muchnik et al.\ \cite{MuchnikSU1998}, no 1-generic sequence is Martin-L\"of random with respect to a computable measure, and thus every 1-generic is weakly deep.  Moreover, as shown by Kautz \cite{Kautz1991}, every 2-random sequence computes a 1-generic, and hence the collection of 1-generics is not negligible.

As the collection of strongly deep sequences is negligible, it is worth asking whether the collection of sequences that are strongly deep with respect to one fixed computable time bound is negligible.  We first introduce some notation.  For a computable time bound $t$ and $c\in\omega$, let $D^t_c(n)=\{X\in\cs\colon K^t(X\uh n) - K(X\uh n)\geq c\}$, which is clopen uniformly in $n$, hence $\bigcap_{n\geq m}D^t_c(n)$ is a $\Pi^0_1$ class.  In addition, we set $D^t_c=\bigcup_{m\in\omega}\bigcap_{n\geq m}D^t_c(n)$.  Then we have:

\begin{thm}
Let $t$ be a computable time bound and $c\in\omega$.  Then $D^t_c$ is not negligible and hence does not consist entirely of strongly deep sequences.
\end{thm}

\begin{proof}
Suppose on the contrary that $D^t_c$ is negligible for some computable time bound $t$ and $c\in\omega$.  Then for each $m\in\omega$, $\bigcap_{n\geq m}D^t_c(n)$ is a negligible $\Pi^0_1$ class.  As shown in \cite[Theorem 5.2]{BienvenuP2016}, no weakly 2-random sequence can compute a member of a negligible $\Pi^0_1$ class, hence no weakly 2-random sequence can compute a member of $\bigcap_{n\geq m}D^t_c(n)$ (recall that $X\in\cs$ is weakly 2-random if $X$ is not contained in any $\Pi^0_2$ of Lebesgue measure zero).  It follows that no weakly 2-random sequence can compute a member of $D^t_c$.  However, there is some weakly 2-random sequence that computes a sequence of high Turing degree \cite{Kautz1991}, and as shown by Juedes, Lathrop, and Lutz \cite{JuedesLL1994}, every high degree contains a strongly deep sequence.  In particular, every high degree contains an element of $D^t_c$, and thus there is some weakly 2-random sequence that computes a member of $D^t_c$, a contradiction.  Note further that under the assumption that $D^t_c$ consists entirely of strongly deep sequences, then by Theorem \ref{thm-deep-neg}, it would follow that $D^t_c$ is negligible, which we have shown cannot hold.
\end{proof}

\section{Variants of Strong Depth}\label{sec-depth-variants}

\subsection{Depth via continuous semimeasures}

In \cite{BienvenuP2016}, the authors show that the definition of a deep $\Pi^0_1$ class can be equivalently defined simply by replacing the universal continuous semimeasure $\M$ with the universal discrete semimeasure $\m$:  a $\Pi^0_1$ class $\P$ with canonical co-c.e. tree $T$ is deep if and only if there is a computable order $h$ such that $\m(T_n) < 2^{-h(n)}$ for all $n\in\omega$, if and only if there is a computable order $f$ such that $\m(T_{f(n)}) < 2^{-n}$ for all $n\in\omega$. Given the characterization of deep sequences in terms of discrete semimeasures in Section \ref{sec-slow-growth}, it is natural to ask whether we can equivalently characterize strong depth by replacing the discrete semimeasures in the definition with continuous semimeasures.  We show that, in this case, we only obtain a characterization of weak depth.

\begin{thm}
$X\in\cs$ is weakly deep if and only if for every computable continuous semimeasure $P$,
\[
\frac{\mathbf{M}(X\uh n)}{P(X\uh n)}\rightarrow\infty.
\]
\end{thm}

\begin{proof}
Suppose that $X$ is not weakly deep.  Then there is some computable measure $\mu$ such that $X$ is $\mu$-Martin-L\"of random.  By the Levin-Schnorr theorem for a priori complexity with respect to the measure $\mu$ (implicit in \cite{Levin1973}), there is some $c$ such that
\[
\mathit{KA}(X\uh n)\geq -\log\mu(X\uh n)-c
\]
for all $n\in\omega$.  Equivalently, we have
\[
\mathbf{M}(X\uh n)\leq 2^c\cdot \mu(X\uh n)
\]
Since every computable measure is a computable semimeasure, the conclusion follows.

For the other direction, suppose that there is some computable, continuous semimeasure and some $c\in\omega$ such that for all $n\in\omega$, 
\begin{equation}\label{eq1}
\frac{\mathbf{M}(X\uh n)}{P(X\uh n)}<c.
\end{equation}
From $P$, we can define a computable measure $\mu$ as follows.  First, we define a function $g:\str\rightarrow\str$ by setting
\[
g(\sigma)=P(\sigma)-(P(\sigma0)+P(\sigma1))
\]
for every $\sigma\in\str$.  Clearly $g$ is computable since $P$ is.  Next we define a computable measure $\mu$ on $\cs$ by setting
\[
\mu(\sigma)=P(\sigma)+\sum_{\tau\prec\sigma}2^{|\tau|-|\sigma|}g(\tau).
\]
One can readily verify that $\mu$ is a measure (see the proof of Proposition 3.5 in \cite{BienvenuHPS2014}) and $P(\sigma)\leq \mu(\sigma)$ for all $\sigma\in\str$.
From Equation (\ref{eq1}) we can derive
\[
\mathit{KA}(X\uh n)\geq -\log P(X\uh n)-O(1)\geq -\log\mu(X\uh n)-O(1).
\]
It thus follows from the Levin-Schnorr theorem for a priori complexity that $X$ is $\mu$-Martin-L\"of random and hence is not weakly deep.
\end{proof}

We define a sequence to be \emph{$\mathit{KA}$-deep} if
\[
\mathit{KA}^t(X \uh n) - \mathit{KA}(X \uh n) \rightarrow \infty
\]
where $\mathit{KA}^t:=-\log \M^t$ as discussed in Section \ref{subsec-isc}.  Recall further from Section \ref{subsec-isc} that for every computable continuous semimeasure $P$, there is some computable time bound $t$ such that $P\leq^\times \mathbf{M}^t$.  It is thus straightforward to show that for $X\in\cs$ and every computable continuous semimeasure $P$,
\[
\frac{\mathbf{M}(X\uh n)}{P(X\uh n)}\rightarrow\infty
\]
if and only if for every computable time bound $t$,
\[
\frac{\mathbf{M}(X\uh n)}{\M^t(X\uh n)}\rightarrow\infty
\]
(see the proof of \cite[Lemma 2.6]{BienvenuDM2023} for discrete semimeasures which directly translates to the case of continuous semimeasures).  Thus we can conclude:

\begin{cor}
$X\in\cs$ is $\mathit{KA}$-deep if and only if $X$ is weakly deep.
\end{cor}

\subsection{Depth and monotone complexity}

We can obtain a similar characterization of weak depth in terms of monotone complexity.  Define a sequence to be \emph{$\mathit{Km}$-deep} if
\[
\mathit{Km}^t(X \uh n) - \mathit{Km}(X \uh n) \rightarrow \infty
\]

This notion was studied by Schnorr and Fuchs \cite{SchnorrF1977}, who used the term \emph{superlearnable} to refer to the failure of being $\mathit{Km}$-deep.  In particular, Schnorr and Fuchs proved that a sequence is superlearnable if and only if it is Martin-L\"of random with respect to a computable measure.    Given that a sequence is weakly deep if and only if it is not Martin-L\"of random with respect to a computable measure, we have the following.

\begin{thm}
$X\in\cs$ is $\mathit{Km}$-deep if and only if $X$ is weakly deep.
\end{thm}

\bibliographystyle{alpha}
\bibliography{bridging_depth}

\newpage
\appendix
\section{Proof of Lemma~\ref{lem-alt-order-deep}}

Working towards the proof of Lemma~\ref{lem-alt-order-deep}, we first establish the following result. 

\begin{lem}
Let $E$ be a computable function which maps every string to a finite set of strings. For every time bound $t$, there exists a time bound~$s$ such that for every $\sigma\in\str$ and every $\tau \in E(\sigma)$: 
\[
K^s(\sigma) + K^s(\tau\mid\sigma) \leq^+ K^t(\sigma,\tau).
\]
\end{lem}

\begin{proof}
Given the computable time bound $t$, $t'$ be a computable time bound such that $2^{-K^t}\leq^\times \mathbf{m}^{t'}$ as guaranteed by Lemma \ref{lem-comp-semi}(iii).
Let $p$ be the conditional discrete semimeasure defined by
\[
p(\tau \mid \sigma) = \frac{\m^{t'}(\sigma,\tau)}{\sum_{\rho \in E(\sigma)} \m^{t'}(\sigma,\rho)} 
\]
when $\tau \in E(\sigma)$ and $p(\tau\mid \sigma)=0$ otherwise (note that a conditional discrete semimeasure $p(\cdot\mid\cdot )$ must satisfy the condition that $p(\cdot\mid\sigma)$ is a discrete semimeasure for each $\sigma\in\str$). It is clear that $p$ is indeed a conditional discrete semimeasure and that it is computable. The denominator in this expression $q(\sigma) := {\sum_{\rho \in E(\sigma)} \m^{t'}(\sigma,\rho)}$ is also a computable discrete semimeasure since $\sum_\sigma q(\sigma) \leq \sum_\sigma \sum_\rho \m^{t'}(\sigma,\rho) \leq \sum_\sigma \sum_\rho \m(\sigma,\rho) \leq 1$. \\

We thus have the identity
\[
\m^{t'}(\sigma,\tau) = p(\tau\mid\sigma) \cdot q(\sigma)
\]
with $p,q$ computable discrete semimeasures. By Lemma~\ref{lem-comp-semi}, there is a computable time bound $s'$ such that $p,q \leq \m^s$ and thus 
\[
\m^{t'}(\sigma,\tau) \leq \m^{s'}(\tau\mid\sigma) \cdot \m^{s'}(\sigma).
\]
Then by our initial assumption on $t$ and  $t'$, we have 
\[
2^{-K^t(\sigma,\tau)}\leq^\times \m^{s'}(\tau\mid \sigma) \cdot \m^{s'}(\sigma).
\]
Finally, by Theorem \ref{thm-comp-semi}, as there is a computable time bound $s$ such that $\mathbf{m}^{s'}\leq^\times 2^{-K^s}$, we have
\[
2^{-K^t(\sigma,\tau)}\leq^\times 2^{-K^s(\tau\mid \sigma)} \cdot 2^{-K^s(\sigma)}.
\]
The lemma follows by taking the negative logarithm of this inequality.
\end{proof}

We can now prove Lemma~\ref{lem-alt-order-deep}. The (i)$\rightarrow$(ii) implication is immediate. For the (ii)$\rightarrow$(i)  implication, suppose $X$ is a sequence and~$h$ a computable increasing function such that 
\[
K^t(X \uh h(n)) - K(X \uh h(n)) \geq^+ n
\]
for any time bound~$t$. If $k \in [h(n), h(n+1))$, we have
\begin{equation}\label{eq:k1}
K(X \uh k) \leq^+ K(X \uh h(n)) + K(X \uh [h(n),k) \mid X \uh h(n)).
\end{equation}

We now apply the previous lemma with $\sigma=X \uh h(n)$, $\tau=X \uh [h(n),k)$ and $E$ the map such that, on input $\rho$, checks whether $\rho$ has length $h(n)$ for some $n$ and if so returns all strings whose length is between $0$ and $h(n+1)-h(n)$ (otherwise $E(\rho)$ is empty). The lemma gives us a computable time bound~$s$ such that

\begin{equation}\label{eq:k2}
K^s(X \uh h(n)) + K^s(X \uh [h(n),k) | X \uh h(n)) \leq^+ K^t(X \uh k).
\end{equation}

Putting~\eqref{eq:k1} and~\eqref{eq:k2} together, and using the fact that $K(X \uh [h(n),k) \mid X \uh h(n)) \leq K^s(X \uh [h(n),k) \mid X \uh h(n))$, it follows that:
\begin{equation}\label{eq:k3}
K^t(X \uh k) - K(X \uh k) \geq^+ K^s(X \uh h(n)) - K(X \uh h(n)) \geq^+ n.
\end{equation}

This being true for all~$n$ and all $k \in [h(n), h(n+1))$, we have
\begin{equation}\label{eq:k4}
K^t(X \uh k) - K(X \uh k) \geq^+ h^{-1}(k),
\end{equation}
and thus $X$ is order-deep. 

Finally, the (ii)$\leftrightarrow$(iii) equivalence can be established using Lemma \ref{lem-comp-semi}(iii) and Theorem \ref{thm-comp-semi}.

\end{document}